\documentclass[11pt,a4paper,oneside]{article}
\usepackage[T1]{fontenc}
\usepackage{lmodern}
\usepackage{tikz}
\usetikzlibrary{arrows.meta,positioning}
\usepackage[a4paper,top=23mm,bottom=25mm,left=28mm,right=28mm]{geometry}
\usepackage{amsmath,amssymb,amsthm,graphicx,needspace}
\usepackage[numbers,sort&compress]{natbib}
\usepackage[unicode,hidelinks]{hyperref}
\hypersetup{pdftitle={Zero-uncertainty states relative to observable algebras},pdfauthor={Jiayu Ran},pdflang={en}}

\numberwithin{equation}{section}
\counterwithin{figure}{section}
\DeclareMathOperator{\Tr}{Tr}
\DeclareMathOperator{\supp}{supp}
\DeclareMathOperator{\MD}{MD}

\DeclareMathOperator{\dist}{dist}
\newcommand{\ket}[1]{\lvert #1\rangle}
\newcommand{\bra}[1]{\langle #1\rvert}
\newcommand{\proj}[1]{\ket{#1}\!\bra{#1}}
\newcommand{\braket}[2]{\langle #1\mid #2\rangle}
\newcommand{\M}{\mathcal M}
\newcommand{\A}{\mathcal A}
\newcommand{\B}{\mathcal B}
\newcommand{\HH}{\mathcal H}
\newcommand{\K}{\mathcal K}
\newcommand{\C}{\mathbb C}

\newcommand{\vct}{\operatorname{vec}}
\theoremstyle{definition}
\newtheorem{theorem}{Theorem}[section]
\newtheorem{corollary}{Corollary}[section]
\newtheorem{lemma}{Lemma}[section]
\newtheorem{proposition}{Proposition}[section]

\allowdisplaybreaks[2]
\begin{document}
\title{Zero-uncertainty states relative to observable algebras}
\author{Jiayu Ran\\[2pt]
  {\normalsize Department of Physics, Fudan University, Shanghai 200433, China}\\
  {\normalsize\href{mailto:ranjiayu@fudan.edu.cn}{\texttt{ranjiayu@fudan.edu.cn}}}}
\date{}
\maketitle

\begin{abstract}
We study zero-uncertainty states (ZUSs) for projective readout with quantum memory. In such tasks, Alice and Bob first share a quantum state. Alice then chooses one measurement from a prescribed family of projection-valued measures (PVMs), performs the measurement, and announces her choice. Bob must infer Alice's outcome from his memory system with zero error. We prove that the admissible shared states are governed by the operator algebra generated by the entire PVM family, and in the physically important case where both systems have equal dimension and the measurements generate the full matrix algebra, every ZUS must be pure and maximally entangled. As a further application to quantum state verification, we use this purity result to give a direct proof that the target maximally entangled state is uniquely verified, extending Zhu's uniqueness criterion to degenerate projective measurements. Moreover, for protocols built from rank-one binary settings, we find an optimal spectral gap of $(3-\sqrt3)/6$ using three settings on two qutrits in $d=3$, and we also characterize all strategies attaining this gap. Finally, we show that different relaxations of exact readout yield inequivalent notions of approximate ZUSs depending on the task.
\end{abstract}

\section{Introduction}
Uncertainty relations constrain how sharply the outcomes of different quantum measurements can be
predicted. Understanding the uncertainty of measurement outcomes is central to quantum information, both
for its bearing on the interpretation of quantum mechanics and for its connection to fundamental sources
of error in quantum circuits and computation. The early work of Heisenberg, Kennard, and Robertson
expressed uncertainty in terms of the variances of
observables~\cite{Heisenberg1927,Kennard1927,Robertson1929}. Maassen--Uffink entropic uncertainty
relations later gave lower bounds on the entropies of outcome distributions for different
measurements~\cite{MaassenUffink1988}. When an observer holds a quantum memory correlated or entangled
with the measured system, conditional entropy quantifies the uncertainty that remains given the memory.
In the uncertainty relation of Berta et al., the lower bound contains the conditional entropy $H(A|B)$,
showing how entanglement with the memory can reduce measurement uncertainty and, in the limiting case,
eliminate it~\cite{Berta2010}. The smooth-entropy framework of Tomamichel and Renner and the review by
Coles et al.\ developed these connections further, including applications to quantum cryptography and
quantum steering~\cite{TomamichelRenner2011,Coles2017}. Zhu introduced and systematically studied
zero-uncertainty states (ZUSs). Once the measurement setting is announced, uncertainty vanishes if the
conditional states of the memory can be distinguished perfectly. Using graph-theoretic methods, Zhu
characterized all ZUSs for families of nondegenerate measurements, identifying the structure of the
shared states for which the memory removes all uncertainty~\cite{Zhu2021}. We study exact readout using
operator algebras and give an application to quantum state verification.

Let $\M=\{\M_x\}_{x\in X}$ be a prescribed family of projection-valued measures (PVMs). Alice performs
one of these measurements and then announces the setting $x$. Bob must identify Alice's outcome using
his quantum memory. A shared state that permits this for every choice is a ZUS for $\M$. For
nondegenerate measurements, Zhu described the structure of these states in terms of a graph of overlaps
between basis vectors and the irreducibility of the basis family~\cite{Zhu2021}. Here we extend this
description to the algebra generated by a general family of projections.

Many readout and verification tasks in quantum information involve degenerate measurements. Earlier work
has also suggested that ZUSs may provide an intrinsic characterization of maximal-entropy structures~\cite{BertaColesWehner2014,Zhu2021}. We
prove that the ZUSs of a family are determined by the operator algebra generated by all of Alice's
outcome projections~\cite{KeylSchlingemannWerner2003}. A coarse-grained measurement may leave part of Alice's system unobserved, yet
several such measurements together can probe every degree of freedom. In that case, degenerate
measurements can impose the same constraints on ZUSs as complete, nondegenerate measurements.

In quantum state verification, these correlations can be used to test a specified maximally
entangled state. When Alice measures in an orthonormal basis and Bob in its complex conjugate,
the target always gives matching outcomes~\cite{HayashiMatsumotoTsuda2006}. For such basis measurements, Zhu characterized when
the target is the only state that always passes~\cite{Zhu2021}, and Zhu and Hayashi constructed optimal verification strategies from mutually unbiased bases~\cite{ZhuHayashi2019PRA1}. We extend the uniqueness result to degenerate measurements. When Alice's outcome projections generate the full matrix algebra and Bob measures their transposes, the target is the only state for which their outcomes always agree. We also obtain the maximal spectral gap for a $d=3$ verification protocol using binary settings.

This paper is organized as follows. We first characterize ZUSs for arbitrary families of PVMs and derive the consequences for equal local dimensions. We then apply this result to verification and determine the optimal spectral gap for three binary settings on two qutrits. Finally, we compare several ways of relaxing the requirement of exact ZUSs.

\section{Exact readout and zero-uncertainty states}
We formulate the readout task in algebraic terms and illustrate the definition of a ZUS with a Bell state.

Let $\HH_A$ be the measured system and $\HH_B$ the memory. Given a family of PVMs $\M=\{\M_x\}_x$ on
$\HH_A$, Alice performs one of the measurements and announces which PVM she chose. Bob must infer
Alice's outcome from his memory $\HH_B$. For a shared state $\rho$, write
$\sigma_{j|x}=\Tr_A[(P_{j|x}\otimes I_B)\rho]$ for his unnormalized conditional state. We call $\rho$ a
ZUS for $\M$, or an $\M$-ZUS, if for every setting $x$ the supports of $\sigma_{j|x}$ and $\sigma_{k|x}$
are orthogonal whenever $j\ne k$. Figure~\ref{fig:readout} summarizes the protocol.
Consider the following example on $\C^2\otimes\C^2$. Let
\begin{align}
\ket{\psi^+}&=\tfrac1{\sqrt 2}(\ket{00}+\ket{11}),\\
\rho_{\rm Bell}&=\proj{\psi^+},\\
\rho_{\rm mix}&=\tfrac12(\proj{00}+\proj{11}).
\end{align}
Let $\mathcal S_1$ consist of the computational-basis and $X$-basis PVMs, and let $\mathcal S_2$ consist of the computational-basis PVM alone. Thus
\begin{align}
\mathcal S_1 &=\bigl\{\{\proj0,\proj1\},\{\proj+,\proj-\}\bigr\},\\
\mathcal S_2 &=\bigl\{\{\proj0,\proj1\}\bigr\}.
\end{align}
The Bell state $\rho_{\rm Bell}$ is an $\mathcal S_1$-ZUS and hence also an $\mathcal S_2$-ZUS. Of these
two families, only $\mathcal S_2$ admits $\rho_{\rm mix}$ as a ZUS. We verify both claims directly.

For each outcome projection $P_\alpha$, Bob's unnormalized conditional operator is
\begin{equation}\sigma_\alpha=\Tr_A[(P_\alpha\otimes I)\rho],\end{equation}
and zero uncertainty requires orthogonal supports for different outcomes of the same PVM. For the Bell state, the computational-basis PVM gives
\begin{align}
\sigma^{(Z)}_0&=\Tr_A[(\proj0\otimes I)\rho_{\rm Bell}]=\tfrac12\proj0,\\
\sigma^{(Z)}_1&=\Tr_A[(\proj1\otimes I)\rho_{\rm Bell}]=\tfrac12\proj1 .
\end{align}
The supports are orthogonal. Likewise, the $X$-basis measurement gives
\begin{align}
\sigma^{(X)}_+&=\Tr_A[(\proj+\otimes I)\rho_{\rm Bell}]=\tfrac12\proj+,\\
\sigma^{(X)}_-&=\Tr_A[(\proj-\otimes I)\rho_{\rm Bell}]=\tfrac12\proj- .
\end{align}
\begin{figure}[t]
\centering
\begin{tikzpicture}[
  stage/.style={draw,rounded corners=1.5pt,fill=black!3,align=center,
                font=\small,inner xsep=4pt,inner ysep=6pt},
  flow/.style={-{Stealth[length=4pt]},semithick},
  node distance=5mm
]
\node[stage,text width=29mm] (share) {Alice and Bob share the state $\rho$.};
\node[stage,text width=31mm,right=of share] (measure)
  {Alice chooses and measures\\$\M_x=\{P_{j|x}\}_j$.};
\node[stage,text width=30mm,right=of measure] (announce)
  {Alice announces $\M_x$ and keeps the outcome $j$ private.};
\node[stage,text width=36mm,right=of announce] (guess)
  {Bob measures $B$ and outputs $\hat{\jmath}$, with $\Pr(\hat{\jmath}=j)=1$.};
\draw[flow] (share) -- (measure);
\draw[flow] (measure) -- (announce);
\draw[flow] (announce) -- (guess);
\end{tikzpicture}
\caption{A shared state $\rho$ is a ZUS for a family of PVMs if Bob can infer Alice's outcome with zero probability of error for every PVM in the family.}
\label{fig:readout}
\end{figure}

The conditional states again have orthogonal supports. The Bell state is therefore a ZUS for both PVMs, hence for the whole family $\mathcal S_1$.

For the mixed state, the computational-basis PVM gives
\begin{align}
\tau^{(Z)}_0&=\Tr_A[(\proj0\otimes I)\rho_{\rm mix}]=\tfrac12\proj0,\\
\tau^{(Z)}_1&=\Tr_A[(\proj1\otimes I)\rho_{\rm mix}]=\tfrac12\proj1 ,
\end{align}
These operators have orthogonal supports, so $\rho_{\rm mix}$ is indeed an $\mathcal S_2$-ZUS. For the $X$-basis measurement,
\begin{align}
\tau^{(X)}_+&=\Tr_A[(\proj+\otimes I)\rho_{\rm mix}]\notag\\
 &=\tfrac12\bigl(|\braket+0|^2\proj0+|\braket+1|^2\proj1\bigr)=\tfrac14 I,\\
\tau^{(X)}_-&=\Tr_A[(\proj-\otimes I)\rho_{\rm mix}]\notag\\
 &=\tfrac12\bigl(|\braket-0|^2\proj0+|\braket-1|^2\proj1\bigr)=\tfrac14 I,
\end{align}
so both conditional operators are proportional to the identity. Their supports coincide, so Bob cannot
distinguish the two outcomes perfectly. Thus $\rho_{\rm mix}$ does not permit zero-error readout of the
$X$-basis result and is not an $\mathcal S_1$-ZUS.

Exact readout also connects ZUSs to other questions in quantum information. In the language of quantum
steering, the conditional memory states $\{\sigma_{j|x}\}_{j,x}$ form an assemblage. The ZUS condition
describes the case in which the outcome-labelled states are perfectly distinguishable once the setting
$x$ is announced~\cite{WisemanJonesDoherty2007}. This suggests connections with joint measurability and
Einstein--Podolsky--Rosen (EPR)
steering~\cite{WisemanJonesDoherty2007,QuintinoVertesiBrunner2014,UolaMoroderGuhne2014,PianiWatrous2015}.
In terms of correctable observables, the same condition says that Bob's memory carries a noiseless
representation of the algebra generated by the tested
projections~\cite{BenyKempfKribs2007,ChoiJohnstonKribs2009}. The structure of ZUSs also has an
application to quantum verification~\cite{Zhu2021}, which we develop below.

Given an arbitrary family $\M$ of PVMs, which states are its ZUSs, and what determines this set? The
main theorem of the next section answers this question. We then give its consequences for several
settings of physical interest.

\section{Structure of ZUSs for arbitrary families of PVMs}\label{sec:structure}
Our main theorem shows that the ZUSs for $\M$ are determined by the algebra generated by its PVMs. We
state the theorem and its principal physical consequence here. Appendix~\ref{app:representation} gives
the proof of the theorem. Section~\ref{s:complete} illustrates its use, and Section~\ref{s:weak}
examines the rigidity of the exact condition.

Let $\M=\{\M_x\}_x$ be a family of PVMs on $\HH_A$, and let $\A$ be the algebra generated by $\M$,
\begin{equation}\A=\A_\M=C^*(\{P_{j|x}\}_{j,x})\subseteq\B(\HH_A).\end{equation}
We use $C^*$-algebra notation throughout the proofs in the appendices. Since the Hilbert spaces are
finite-dimensional, $\A$ can be viewed as a matrix algebra.

For a bipartite state $\rho$, define
\begin{gather}
\Lambda_\rho(Y)=\Tr_A[(Y^T\otimes I_B)\rho],\\
\rho_B=\Lambda_\rho(I_A),\\
\K_\rho=\supp(\rho_B).
\end{gather}
The map $Y\mapsto\Lambda_\rho(Y)$ is completely positive (CP)~\cite{Choi1975}. The superscript $T$
denotes transposition in the standard orthonormal basis throughout. In finite dimensions, if
$\mathcal R^T$ denotes the family of transposes of a family of operators $\mathcal R$, then
$C^*(\mathcal R^T)\simeq C^*(\mathcal R)$~\cite{Davidson1996}. The transpose arises from the
Choi--Jamio{\l}kowski convention~\cite{Choi1975,Jamiolkowski1972}.

Let $\A$ be a unital $C^*$-algebra, let $\HH$ be a Hilbert space, and let $\B(\HH)$ denote the algebra of bounded linear operators on $\HH$. A map
\begin{equation}\pi:\A\longrightarrow\B(\HH)\end{equation}
is a unital $*$-representation if it is an algebra homomorphism,
\begin{equation}\pi(\lambda a+\mu b)=\lambda\pi(a)+\mu\pi(b),\qquad \pi(ab)=\pi(a)\pi(b),\end{equation}
preserves adjoints,
\begin{equation}\pi(a^*)=\pi(a)^*,\end{equation}
and preserves the identity,
\begin{equation}\pi(1_{\A})=I_{\HH},\end{equation}
where $1_\A$ is the identity of $\A$ and $I_{\HH}$ is the identity operator on $\HH$. The image of a
$C^*$-algebra under a unital $*$-representation is again a $C^*$-algebra~\cite{Davidson1996}.

Given a subset $\mathcal T$ of an algebra $\A$, its commutant in $\A$ is
\begin{equation}\mathcal T'=\{a\in\A:at=ta,\ \forall t\in\mathcal T\}.\end{equation}
Thus $\mathcal T'$ consists of all elements of $\A$ that commute with every element of $\mathcal T$. The
commutation requirement is relative to $\mathcal T$. The commutant is itself a subalgebra, and its
elements need not commute with one another.

For a family of PVMs $\M=\{\M_x\}_x$ on $\HH_A$, set $\A=\A_\M=C^*(\{P_{j|x}\}_{j,x})$.
\begin{theorem}[Algebraic structure of ZUSs]\label{thm:representation}
A bipartite state $\rho$ is a ZUS for $\M$ if and only if the restriction of its normalized map to $\A^T$, given by
\begin{equation}\Phi_\rho(Y):=\rho_B^{-1/2}\Lambda_\rho(Y)\rho_B^{-1/2},\qquad Y\in\A^T,\end{equation}
and denoted by $\pi_\rho:=\Phi_\rho|_{\A^T}$, satisfies both of the following conditions.
\begin{enumerate}
\item $\pi_\rho:\A^T\to\B(\K_\rho)$ is a unital $*$-representation.
\item $\rho_B\in[\pi_\rho(\A^T)]'$.
\end{enumerate}
\end{theorem}
We prove the theorem in Appendix~\ref{app:representation}. Its principal consequence for systems of
equal dimension is the following corollary, proved in Appendix~\ref{app:rigidity}.
\begin{corollary}\label{cor:rigid}
If $\A_\M=\B(\HH_A)$ and $\dim\HH_A=\dim\HH_B=d$, then every ZUS for $\M$ is maximally entangled up to local unitary equivalence,
\begin{equation}\rho\simeq\proj{\psi_d},\qquad \ket{\psi_d}=\frac{1}{\sqrt d}\sum_{i=0}^{d-1}\ket{ii}.\end{equation}
\end{corollary}
Theorem~\ref{thm:representation} gives an algebraic generalization of Zhu's Theorem~1~\cite{Zhu2021}.

Theorem~\ref{thm:representation} and its corollary show how the constraints on ZUSs depend on the
algebra generated by the whole family, even when individual measurements are degenerate. The first
example illustrates this dependence. The second shows that arbitrarily weak noncommutativity can still
force maximal entanglement.

\subsection{Example 1. Degenerate measurements that generate the full algebra}\label{s:complete}
To illustrate Theorem~\ref{thm:representation}, suppose Alice's system consists of two qubits $A_1$ and $A_2$. Define three rank-two projections by
\begin{align}
P_Z&=\proj0_{A_1}\otimes I_{A_2},\\
P_X&=\proj+_{A_1}\otimes I_{A_2},\\
P_R&=\proj0_{A_1}\otimes\proj+_{A_2}+\proj1_{A_1}\otimes\proj0_{A_2}.
\end{align}
Each binary PVM $\{P,I-P\}$ is degenerate. The first two projections fully probe the first qubit $A_1$
and act trivially on the second. The third couples the outcome on $A_1$ to two incompatible tests on
$A_2$, namely $\proj+$ and $\proj0$. Together the three projections generate
\begin{equation}C^*(P_Z,P_X,P_R)=\B(\C^2\otimes\C^2),\end{equation}
as shown in Lemma~\ref{lem:twoqubit}. By the full-algebra consequence of
Theorem~\ref{thm:representation}, if Bob's memory is four-dimensional, every ZUS for this family is
maximally entangled between Alice's four-dimensional system $\C^2\otimes\C^2$ and Bob's memory.

\subsection{Example 2. Weak noncommutativity}\label{s:weak}
A second example illustrates the rigidity of exact readout. On a qubit let $P=\proj0$ and
$Q_\theta=\proj{0_\theta}$, where $\ket{0_\theta}=\cos\theta\,\ket0+\sin\theta\,\ket1$ and
$0<\theta\ll1$. Direct calculation gives $\|[P,Q_\theta]\|=|\sin(2\theta)|/2\ll1$, which can be made
arbitrarily small. These projections nevertheless generate the full algebra $\B(\C^2)$, as shown in
Lemma~\ref{lem:tilted}. Thus, in the equal-dimensional setting of the corollary, even an arbitrarily
small nonzero tilt $\theta$ forces every ZUS for this family to be maximally entangled.

We next apply the structure theorem to quantum verification.

\section{Application to quantum state verification}\label{sec:verification}

In this section, as an application of Theorem~\ref{thm:representation}, we consider a specified verification protocol of a given 
maximally entangled state $\ket{\psi_d}=d^{-1/2}\sum_{a=0}^{d-1}\ket{aa}$ on two $d$-dimensional 
systems. In each round, setting $x$ is chosen with probability $\mu_x>0$, with $\sum_x\mu_x=1$. Alice
measures the PVM $\{P_{j|x}\}_j$, and Bob measures its transpose $\{P_{j|x}^T\}_j$, where the transpose is defined through the basis of our given maximally entangled state. The test passes when their outcome labels agree. The pass operator for one
setting and the final average verification operator are, respectively,
\begin{equation}\label{eq:strategy}
\Omega_x=\sum_j P_{j|x}\otimes P_{j|x}^T,\qquad
\Omega=\sum_x\mu_x\Omega_x .
\end{equation}
For fixed $x$, the operators $P_{j|x}\otimes P_{j|x}^T$ are mutually orthogonal projections, so $\Omega_x$ is also an orthogonal projection. Using $P_{j|x}^2=P_{j|x}$, $\sum_jP_{j|x}=I$, and the transpose trick of  maximally entangled state $(I\otimes P^T)\ket{\psi_d}=(P\otimes I)\ket{\psi_d}$,
we find
\begin{subequations}\label{eq:target-pass}
\begin{align}
\Omega_x\ket{\psi_d}
 &=\sum_j(P_{j|x}^2\otimes I)\ket{\psi_d}=\ket{\psi_d},\\
\Tr\!\left(\Omega\proj{\psi_d}\right)
 &=\sum_x\mu_x\bra{\psi_d}\Omega_x\ket{\psi_d}=\sum_x\mu_x=1 .
\end{align}
\end{subequations}
The target $\ket{\psi_d}$ therefore passes every setting with probability $1$. Thus, the following theorem holds directly.

\begin{theorem}[Unique verification of  maximally entangled state]\label{thm:verify}
Let $\A=C^*(P_{j|x}:j,x)$. If $\A=M_d(\C)$, then $\Tr(\Omega\rho)=1$ if and only if $\rho=\proj{\psi_d}$.
\end{theorem}
\begin{proof}
\textbf{Sufficiency.} If $\rho=\proj{\psi_d}$, then $\Tr(\Omega\rho)=1$ follows directly from Eq.~\eqref{eq:target-pass}.

\textbf{Necessity.} Suppose $\Tr(\Omega\rho)=1$ and set $\tau=(\rho+\proj{\psi_d})/2$. By linearity of the trace and Eq.~\eqref{eq:target-pass}, we have
\begin{equation}
  \Tr(\Omega\tau)=\tfrac12\Tr(\Omega\rho)+\tfrac12\bra{\psi_d}\Omega\ket{\psi_d}=1.
\end{equation}
Since every $\mu_x$ is positive and each $\Tr(\Omega_x\tau)\le 1$, their weighted average can equal one only if
$\Tr(\Omega_x\tau)=1$ for every $x$. Thus $\tau$ is a ZUS for this family. Since $\A=M_d(\C)$ and the local dimensions are equal, Corollary~\ref{cor:rigid} implies that
$\tau$ is pure. Hence, its decomposition tells us that the supports of both $\rho$ and $\proj{\psi_d}$ lie in the one-dimensional support of
$\tau$. Therefore, $\rho=\proj{\psi_d}$.
\end{proof}

This extends the uniqueness criterion in Zhu's Theorem~3 to degenerate projective measurements~\cite{Zhu2021}. The conclusion also follows from
Kribs's Theorem~2.1~\cite{Kribs2003Fixed}, for which Holbrook, Kribs, and Laflamme gave a shorter proof~\cite{HolbrookKribsLaflamme2003}. Here Theorem~\ref{thm:representation} gives a direct proof of the verification criterion.

We can also obtain the following single-round fidelity bound implied by Eq.~(1) of Zhu and
Hayashi~\cite{ZhuHayashi2019PRA1}. See also
Refs.~\cite{ZhuHayashi2019PRA2,PallisterLindenMontanaro2018}.

\begin{corollary}[Fidelity bound]\label{cor:fidelity}
Suppose $\ket{\psi_d}$ is the unique pure state that passes with probability $1$. Write
$\nu=1-\lambda_2(\Omega)>0$, where $\lambda_2$ is the second-largest eigenvalue counted with
multiplicity. If $\Tr(\Omega\rho)\ge1-\epsilon$, then
\begin{equation}
\bra{\psi_d}\rho\ket{\psi_d}\ge1-\epsilon/\nu .
\end{equation}
\end{corollary}
\begin{proof}
Write $F=\bra{\psi_d}\rho\ket{\psi_d}$. On the orthogonal complement of the target, $\Omega\le(1-\nu)I$.
Thus $\Tr(\Omega\rho)\le F+(1-\nu)(1-F)=1-\nu(1-F)$, which gives the bound.
\end{proof}

We then consider a verification protocol of a specified two-qutrit maximally entangled state $\ket{\psi_3}$.   We require
$\proj{\psi_3}$ to be the only input state that passes with probability $1$. Alice and Bob use three binary settings. In each round they choose $i\in\{1,2,3\}$ with probability $\mu_i$, where $\mu_i>0$ and $\sum_{i=1}^3\mu_i=1$. In setting $i$, Alice measures $\{Q_i,I-Q_i\}$, where each $Q_i=\proj{q_i}$ is a rank-one orthogonal projection, and Bob measures $\{Q_i^T,(I-Q_i)^T\}$. The round
passes when the two outcome labels agree.
The verification operator is
\begin{equation}\label{eq:qutrit-strategy}
\Omega=\sum_x\mu_x\Omega_x,\qquad \Omega_x=\sum_j Q_{j|x}\otimes Q_{j|x}^T.
\end{equation}
We maximize the spectral gap $\nu(\Omega)=1-\lambda_2(\Omega)$ over the three rank-one projectors and their probabilities.
We claim the following theorem.
\begin{theorem}\label{maxgap}
  The maximal spectral gap over these strategies is

\begin{equation}\label{eq:qutritgap}
\nu_{\mathrm{opt}}(3)=\frac{3-\sqrt3}{6}.
\end{equation}
One choice is $Q_i=\proj{q_i}$ with $\mu_i=1/3$, where $\omega=e^{2\pi\mathrm i/3}$ and
\begin{equation}\label{eq:qutrit-design}
q_1=\frac{(1,1,0)^T}{\sqrt2},\qquad
q_2=\frac{(0,1,1)^T}{\sqrt2},\qquad
q_3=\frac{(0,1,\omega)^T}{\sqrt2}.
\end{equation}
\end{theorem}
We will prove it in Appendix~\ref{app:qutrit} and discuss the general choice attaining the gap in theorem~\ref{innerp}.

For any given $d$ and $k$, the algebraic criterion also helps us to determine  the minimum number of settings
needed to verify $\ket{\psi_d}$ uniquely. The detailed analytical solution can be derived from Proposition~2.1 and Theorem~1.1 of
Wright~\cite{Wright2011}. But for a given family of projections, one can check whether the target is
uniquely verified using the result
\begin{equation}\label{eq:binarycriterion}
\begin{gathered}
\text{only }\proj{\psi_d}\text{ always passes}\\
\Longleftrightarrow\quad C^*(Q_x:x)=M_d(\C)\\
\Longleftrightarrow\quad \{Q_x\}_x'=\C I,
\end{gathered}
\end{equation}
which is usually straightforward to check in practice.

The last equivalence follows from the finite-dimensional double commutant theorem~\cite{Davidson1996}.
If the projections generate $M_d(\C)$, uniqueness follows from Theorem~\ref{thm:verify}.
Conversely, if they generate a proper subalgebra, their commutant contains a nontrivial orthogonal
projection $R$. Since $R\otimes I$ commutes with every $\Omega_x$, the vector
$(R\otimes I)\ket{\psi_d}$ is also fixed by every $\Omega_x$. Its Schmidt rank is
$\operatorname{rank}R<d$, so after normalization it gives a second state that passes with certainty.

\section{Task dependence of approximate zero uncertainty}\label{sec:approx}
The structure theorem characterizes exact ZUSs algebraically. A two-qubit example shows how different
ways of allowing errors lead to different notions of approximation. We first compare the Bayes error
with the error of a hypothesis test for a declared outcome. We then consider the error in a single round
of a verification test used for Bell-state traps, followed by a measure of departure from the algebraic
conditions.

Let
\begin{equation}\ket{\psi_\lambda}=\sqrt\lambda\ket{00}+\sqrt{1-\lambda}\ket{11},\qquad 0\le\lambda\le1.\end{equation}
The choice $\lambda=1/2$ gives a Bell state, and $\lambda=0,1$ give product states. The task uses the
computational-basis PVM and a PVM in the tilted basis
\begin{align}
\ket{0_\theta}&=\cos\theta\,\ket0+\sin\theta\,\ket1,\\
\ket{1_\theta}&=-\sin\theta\,\ket0+\cos\theta\,\ket1,\quad 0<\theta<\pi/4 .
\end{align}
For this range of $\theta$, the two PVMs generate the full qubit algebra $M_2(\C)$.

For the tilted basis, the two unnormalized conditional vectors in Bob's memory are
$\ket{\beta_0}=\sqrt\lambda\cos\theta\ket0+\sqrt{1-\lambda}\sin\theta\ket1$ and
$\ket{\beta_1}=-\sqrt\lambda\sin\theta\ket0+\sqrt{1-\lambda}\cos\theta\ket1$. Denote their unnormalized
density operators by $\sigma_0$ and $\sigma_1$. The normalized states are
$\hat\sigma_0=\sigma_0/\Tr\sigma_0$ and $\hat\sigma_1=\sigma_1/\Tr\sigma_1$.

\subsection{Bayes error}
One natural statistic is the Bayes error. For $\lambda\ne1/2$, the state $\ket{\psi_\lambda}$ does not
allow Bob to distinguish the outcomes of the tilted-basis measurement perfectly. The Helstrom bound
gives the minimum probability of error~\cite{Helstrom1976,YuenKennedyLax1975},
\begin{equation} e_{\rm Bayes}=\tfrac12(1-\|\sigma_0-\sigma_1\|_1).\end{equation}
For these two unnormalized rank-one states, direct calculation gives
\begin{gather}
\|\sigma_0-\sigma_1\|_1=\sqrt{1-4|\braket{\beta_0}{\beta_1}|^2},\\
\braket{\beta_0}{\beta_1}=(1-2\lambda)\sin\theta\cos\theta ,
\end{gather}
Substitution yields
\begin{equation}
e_{\rm Bayes}(\lambda,\theta)=\tfrac12\left[1-\sqrt{1-(1-2\lambda)^2\sin^2(2\theta)}\right].
\end{equation}
The error vanishes at the exact ZUS point $\lambda=1/2$. It also tends to zero for $\lambda=1$ as
$\theta$ tends to zero. Thus a small Bayes error can occur even when the shared state is far from every
exact ZUS. In particular, there is no function $f(\epsilon)$ tending to zero with $\epsilon$ that bounds
the trace distance to the corresponding ZUS set for all these tasks whenever the Bayes error is at most
$\epsilon$. Proposition~\ref{prop:bayes} gives the precise statement and proof.

\subsection{Hypothesis testing for a declared outcome}
Another statistical task conditions on a declared outcome. Suppose, for example, that the declaration is
$0_\theta$. The two hypotheses for Bob's normalized conditional state are $\hat\sigma_0$ and
$\hat\sigma_1$. Giving the hypotheses equal prior probabilities removes any advantage due to the rarity
of one of the outcomes.

For the tilted basis, write
\begin{equation}
c_{\lambda,\theta}=\Tr(\hat\sigma_0\hat\sigma_1)=\frac{(1-2\lambda)^2\sin^2(2\theta)}{1-(1-2\lambda)^2\cos^2(2\theta)} .
\end{equation}
The Helstrom error for this binary test with equal priors~\cite{Helstrom1976,YuenKennedyLax1975} is
\begin{equation} e_{\rm hyp}(\lambda,\theta)=\tfrac12\bigl[1-\sqrt{1-c_{\lambda,\theta}}\,\bigr]. \end{equation}
In particular, $e_{\rm hyp}(1,\theta)=1/2$ and $e_{\rm Bayes}(1,\theta)=\sin^2\theta$.

\subsection{Error in a single-round trap test}
Consider the matching-outcome check used for a Bell-state trap. Suppose Alice performs the tilted PVM,
obtains $0_\theta$, and reports it truthfully. Bob then measures his memory with the same tilted PVM. He
accepts if he also obtains $0_\theta$ and rejects otherwise. Conditional on Alice's outcome, the
probability of rejecting an honest report is
\begin{equation}
\begin{split}
e_{\rm false}&=1-\Tr\bigl(\proj{0_\theta}\hat\sigma_0\bigr)\\
&=\frac{\sin^2\theta\cos^2\theta\,(\sqrt\lambda-\sqrt{1-\lambda})^2}
{\lambda\cos^2\theta+(1-\lambda)\sin^2\theta}.
\end{split}
\end{equation}
This probability vanishes at the Bell point $\lambda=1/2$ and is generally positive away from it. Thus
Bob's fixed matching-outcome check can reject an honest report with nonzero probability away from the
exact ZUS point.

This error probability depends on Bob's choice of test. A different decision rule, such as a Helstrom
minimum-error test~\cite{Helstrom1976,YuenKennedyLax1975} or a Neyman--Pearson
test~\cite{NeymanPearson1933}, gives a different error criterion.

\subsection{Departure from the algebraic conditions}
The same two-qubit setting gives an example in which the unital $*$-representation condition of
Theorem~\ref{thm:representation} holds while the commutation condition fails.

Take
\begin{equation}\ket{\psi_\lambda}=\sqrt\lambda\ket{00}+\sqrt{1-\lambda}\ket{11},\qquad 0<\lambda<1,\quad\lambda\ne1/2,\end{equation}
and set $\rho=\proj{\psi_\lambda}$. The two Schmidt coefficients differ, so $\rho$ is not maximally entangled. Write
\begin{equation}D_\lambda=\begin{pmatrix}\sqrt\lambda&0\\0&\sqrt{1-\lambda}\end{pmatrix}.\end{equation}
We show that $\Phi_\rho$ is a unital $*$-representation and that $\rho_B$ lies outside its commutant.
\begin{proof}
For any matrix
\begin{equation}Y=\begin{pmatrix}a&b\\c&d\end{pmatrix}\in M_2(\C),\end{equation}
the definition gives
\begin{equation}
\begin{split}
\Lambda_\rho(Y)&=\Tr_A[(Y^T\otimes I)\rho]\\
&=\begin{pmatrix}\lambda a&\sqrt{\lambda(1-\lambda)}\,b\\\sqrt{\lambda(1-\lambda)}\,c&(1-\lambda)d\end{pmatrix}.
\end{split}
\end{equation}
Equivalently,
\begin{equation}\Lambda_\rho(Y)=D_\lambda YD_\lambda.\end{equation}
We also have
\begin{equation}\rho_B=\Lambda_\rho(I)=D_\lambda^2,\qquad \rho_B^{-1/2}=D_\lambda^{-1}.\end{equation}
It follows that
\begin{equation}\Phi_\rho(Y)=\rho_B^{-1/2}\Lambda_\rho(Y)\rho_B^{-1/2}=D_\lambda^{-1}(D_\lambda YD_\lambda)D_\lambda^{-1}=Y.\end{equation}
Thus $\Phi_\rho$ is the identity map on $M_2(\C)$, and $\pi_\rho=\Phi_\rho$ is a unital $*$-representation.

Since $\pi_\rho$ is the identity representation, $\pi_\rho(M_2(\C))=M_2(\C)$. Membership in the
commutant would require $\rho_B$ to commute with every $2\times2$ matrix. But the matrix
\begin{equation}\rho_B=\begin{pmatrix}\lambda&0\\0&1-\lambda\end{pmatrix}\end{equation}
has unequal diagonal entries for $\lambda\ne1/2$, so it fails this condition.
\end{proof}

The operator norm quantifies this departure. The commutant $[M_2(\C)]'$ consists of scalar multiples of the identity, so
\begin{equation}
\dist\bigl(\rho_B,[M_2(\C)]'\bigr)=\inf_{z\in\C}\|\rho_B-zI\|
=\left|\lambda-\tfrac12\right|=\tfrac12|1-2\lambda|.
\end{equation}
This distance from $\rho_B$ to the scalar matrices vanishes precisely at $\lambda=1/2$.

The first condition of Theorem~\ref{thm:representation} is related to correctable observables in the
Heisenberg picture, since $\pi_\rho$ preserves projections, orthogonality, and
products~\cite{BenyKempfKribs2007,ChoiJohnstonKribs2009}. Exact readout also requires that this
algebraic structure allow Bob to determine Alice's outcome. This is the additional condition that
$\rho_B$ commute with $\pi_\rho(\A^T)$.

\subsection{Comparison}
These examples give different notions of approximate zero uncertainty. Different statistical assumptions
or information-processing tasks relax different constraints and thus lead to different robustness
statements. The appropriate approximation depends on the task.

\section{Conclusion}
The algebra generated by the outcome projections of a measurement family determines its
zero-uncertainty states. This extends Zhu's characterization to degenerate projective
measurements~\cite{Zhu2021}. When the local dimensions are equal and the projections generate
the full matrix algebra, every such state is pure and maximally entangled.

This rigidity gives a direct proof of unique verification of a maximally entangled state
using mutually transposed measurements. For two qutrits, we determined the optimal spectral
gap $(3-\sqrt3)/6$ for three rank-one binary settings and characterized all strategies that
attain it. How our main results can help us understand other quantum-information scenarios remains to be explored.

\section*{Acknowledgments}
I am grateful to Professor Huangjun Zhu for his support. His work on zero-uncertainty states with quantum memory inspired this study.

I thank Ma-ke Yuan, Zhaoxuan Bian, Yuhang Wu, and Yukun Wu for helpful discussions during the preparation of this paper.

Google's Gemini 3 Pro helped me find the crucial multiplicative-domain lemma (Lemma~\ref{lem:md}), which enabled me to complete the proof of Theorem~\ref{thm:representation}. I used OpenAI's GPT-6 Astra and GPT-6.1 Sol to substantially shorten the proofs of Theorems~\ref{maxgap} and~\ref{innerp}.

I wrote the original manuscript in Chinese. GPT-6 Astra did nearly all of the English translation and typesetting for \textit{Physical Review A}.

I used OpenAI's ChatGPT (March 2026) to search for and organize the relevant literature, and I used these tools to review every derivation. I personally checked all AI-assisted results in full and take full responsibility for the entire manuscript.

\appendix

\section{Proof of Theorem~\ref{thm:representation}}\label{app:representation}
We use the standard multiplicative-domain lemma for unital completely positive maps.

Let $\Phi:M_d(\C)\to\B(\HH)$ be a finite-dimensional unital CP map. Its multiplicative domain is
\begin{equation}
\begin{split}
\MD(\Phi)&=\{X:\Phi(X^*X)=\Phi(X)^*\Phi(X),\\
&\hphantom{=\{X:}\Phi(XX^*)=\Phi(X)\Phi(X)^*\},
\end{split}
\end{equation}
The following facts are Theorem~3.1 of Ref.~\cite{Choi1974Schwarz} and Theorem~8 of Ref.~\cite{ChoiJohnstonKribs2009}.
\begin{lemma}[Multiplicative-domain lemma]\label{lem:md}
\leavevmode
\begin{enumerate}
\item $\MD(\Phi)$ is a unital $C^*$-subalgebra.
\item The restriction $\Phi|_{\MD(\Phi)}$ is a unital $*$-representation.
\item If $P$ and $\Phi(P)$ are both projections, then $P\in\MD(\Phi)$.
\end{enumerate}
\end{lemma}
\begin{proof}[Proof of Theorem~\ref{thm:representation}]
We prove necessity and sufficiency separately.

\textbf{Necessity. Constructing the representation from exact readout.}
Suppose $\rho$ is a ZUS for $\M$. Write $\K_\rho=\supp(\rho_B)$. The operator $\rho_B$ is strictly
positive on $\K_\rho$, so $\rho_B^{-1/2}$ is well defined there. Set
\begin{equation}
\Phi_\rho(Y)=\rho_B^{-1/2}\Lambda_\rho(Y)\rho_B^{-1/2},\qquad Y\in\B(\HH_A).
\end{equation}
The map $\Lambda_\rho$ is completely positive, so $\Phi_\rho$ is also completely positive. Moreover,
\begin{equation}\Phi_\rho(I_A)=\rho_B^{-1/2}\rho_B\rho_B^{-1/2}=I_{\K_\rho}.\end{equation}
Thus $\Phi_\rho$ is unital.

Define
\begin{equation}\pi_\rho:=\Phi_\rho|_{\A^T}.\end{equation}
For every $Y\in\A^T$, this gives
\begin{equation}\pi_\rho(Y)=\rho_B^{-1/2}\Lambda_\rho(Y)\rho_B^{-1/2}.\end{equation}
We show that $\pi_\rho$ is a unital $*$-representation.

By Lemma~\ref{lem:md}, it suffices to prove that the domain of the restriction $\pi_\rho=\Phi_\rho|_{\A^T}$ satisfies
\begin{equation}\A^T\subseteq\MD(\Phi_\rho).\end{equation}
Since
\begin{equation}\A^T=C^*(\{P_{j|x}^T\}_{j,x}),\end{equation}
and $\MD(\Phi_\rho)$ is a unital $C^*$-subalgebra, it is enough to show that each generator $P_{j|x}^T$ belongs to $\MD(\Phi_\rho)$.

Fix $x$ and $j$. By definition,
\begin{equation}
\sigma_{j|x}=\Lambda_\rho(P_{j|x}^T)=\Tr_A[(P_{j|x}\otimes I_B)\rho]
\end{equation}
is the unnormalized conditional state. For fixed $x$, the normalization relation is
\begin{equation}\label{eq:asum}\sum_j\sigma_{j|x}=\Lambda_\rho(I_A)=\rho_B.\end{equation}
The ZUS condition gives orthogonal supports for the conditional states associated with distinct outcomes
of the same PVM. Hence $\sigma_{j|x}\sigma_{k|x}=0$ for $j\ne k$. It follows that
\begin{align}
\rho_B\sigma_{j|x}&=\Bigl(\sum_k\sigma_{k|x}\Bigr)\sigma_{j|x}=\sigma_{j|x}^2,\\
\sigma_{j|x}\rho_B&=\sigma_{j|x}\Bigl(\sum_k\sigma_{k|x}\Bigr)=\sigma_{j|x}^2 .
\end{align}
Thus $\sigma_{j|x}$ commutes with $\rho_B$, and
\begin{equation}\label{eq:acomm}\rho_B\sigma_{j|x}=\sigma_{j|x}\rho_B=\sigma_{j|x}^2.\end{equation}
Equation~\eqref{eq:asum} also gives $\sigma_{j|x}\le\rho_B$, so the support of $\sigma_{j|x}$ lies in
$\K_\rho=\supp(\rho_B)$. We may therefore define
\begin{equation}R_{j|x}:=\Phi_\rho(P_{j|x}^T)=\rho_B^{-1/2}\sigma_{j|x}\rho_B^{-1/2}.\end{equation}
We claim that $R_{j|x}$ is a projection. First, the expression $\rho_B^{-1/2}\sigma_{j|x}\rho_B^{-1/2}$
is self-adjoint because both factors $\rho_B^{-1/2}$ and $\sigma_{j|x}$ are self-adjoint. By
Eq.~\eqref{eq:acomm}, $\sigma_{j|x}$ commutes with $\rho_B$ and hence with $\rho_B^{-1}$ and
$\rho_B^{-1/2}$ on $\K_\rho$. On this support we have
\begin{equation} R_{j|x}=\rho_B^{-1}\sigma_{j|x}.\end{equation}
Using Eq.~\eqref{eq:acomm} again gives
\begin{equation}
R_{j|x}^2=\rho_B^{-2}\sigma_{j|x}^2=\rho_B^{-2}\rho_B\sigma_{j|x}=\rho_B^{-1}\sigma_{j|x}=R_{j|x}.
\end{equation}
This proves that $R_{j|x}$ is a projection.

Both $P_{j|x}^T$ and $\Phi_\rho(P_{j|x}^T)=R_{j|x}$ are projections. Lemma~\ref{lem:md} therefore gives
$P_{j|x}^T\in\MD(\Phi_\rho)$ for every $j,x$. Consequently,
\begin{equation}\A^T=C^*(\{P_{j|x}^T\}_{j,x})\subseteq\MD(\Phi_\rho).\end{equation}
The restriction of $\Phi_\rho$ to its multiplicative domain is a unital $*$-homomorphism by
Lemma~\ref{lem:md}. Its restriction $\pi_\rho=\Phi_\rho|_{\A^T}$ is therefore a unital
$*$-representation.

It remains to prove the commutation condition. We have $R_{j|x}=\rho_B^{-1}\sigma_{j|x}$, and both
$\rho_B^{-1}$ and $\sigma_{j|x}$ commute with $\rho_B$. Thus
\begin{equation}[R_{j|x},\rho_B]=0.\end{equation}
Since
\begin{equation}\pi_\rho(P_{j|x}^T)=R_{j|x},\end{equation}
the operator $\rho_B$ commutes with the image $\pi_\rho(P_{j|x}^T)$ of each generator. These images
generate the $C^*$-algebra $\pi_\rho(\A^T)$, so $\rho_B$ commutes with the entire image. In other words,
\begin{equation}\rho_B\in[\pi_\rho(\A^T)]'.\end{equation}
This proves necessity.

\textbf{Sufficiency. Deducing exact readout from the two conditions.}
Conversely, suppose that $\pi_\rho:=\Phi_\rho|_{\A^T}$ satisfies the two conditions of the theorem. To
show that $\rho$ is a ZUS for $\M$, we prove the stronger statement that it is a ZUS for every PVM in
$\A$.

Take any PVM $\{E_\alpha\}_\alpha\subseteq\A$ and set
\begin{equation}Q_\alpha:=\pi_\rho(E_\alpha^T).\end{equation}
Since $\pi_\rho$ is a unital $*$-representation, $\{Q_\alpha\}_\alpha$ is a PVM on $\K_\rho$. In particular, for $\alpha\ne\beta$,
\begin{equation}Q_\alpha Q_\beta=\pi_\rho(E_\alpha^T)\pi_\rho(E_\beta^T)=\pi_\rho(E_\alpha^T E_\beta^T)=0.\end{equation}
When Alice measures $\{E_\alpha\}_\alpha$, Bob's unnormalized conditional state for outcome $\alpha$ is
\begin{equation}\sigma_\alpha:=\Tr_A[(E_\alpha\otimes I_B)\rho]=\Lambda_\rho(E_\alpha^T).\end{equation}
The definition of $\pi_\rho=\Phi_\rho|_{\A^T}$ gives
\begin{equation}\sigma_\alpha=\rho_B^{1/2}Q_\alpha\rho_B^{1/2}.\end{equation}
The commutation condition says that $\rho_B$ commutes with each $Q_\alpha$. Hence
\begin{equation}\sigma_\alpha=\rho_B^{1/2}Q_\alpha\rho_B^{1/2}=Q_\alpha\rho_B=\rho_BQ_\alpha.\end{equation}
For $\alpha\ne\beta$ we obtain
\begin{equation}\sigma_\alpha\sigma_\beta=\rho_B Q_\alpha Q_\beta\rho_B=0,\end{equation}
so that
\begin{equation}\supp(\sigma_\alpha)\perp\supp(\sigma_\beta),\qquad\alpha\ne\beta .\end{equation}
This proves sufficiency.
\end{proof}

\section{Proof of the equal-dimension corollary to Theorem~\ref{thm:representation}}\label{app:rigidity}
\begin{proof}[Proof of Corollary~\ref{cor:rigid}]
Assume that
\begin{equation}\A_\M=\B(\HH_A),\qquad \dim\HH_A=\dim\HH_B=d,\end{equation}
and let $\rho$ be an $\M$-ZUS. By Theorem~\ref{thm:representation}, the map
\begin{equation}\pi_\rho:\A_\M^T\longrightarrow\B(\K_\rho)\end{equation}
is a unital $*$-representation, and
\begin{equation}\rho_B\in[\pi_\rho(\A_\M^T)]'.\end{equation}
The full-algebra condition gives $\A_\M=\B(\HH_A)$. Transposing the full matrix algebra leaves it invariant, so
\begin{equation}\A_\M^T\simeq M_d(\C).\end{equation}
Every nonzero finite-dimensional unital $*$-representation of $M_d(\C)$ is a direct sum of copies of its
standard representation~\cite{Davidson1996}. It follows that for some positive integer $n$ the
representation space $\K_\rho$ has dimension
\begin{equation}\dim\K_\rho=nd.\end{equation}
On the other hand,
\begin{equation}\K_\rho=\supp(\rho_B)\subseteq\HH_B,\end{equation}
which gives
\begin{equation} d\le\dim\K_\rho\le\dim\HH_B=d.\end{equation}
Thus
\begin{equation}\dim\K_\rho=\dim\HH_A,\qquad \K_\rho=\HH_B\simeq\HH_A.\end{equation}
The representation has multiplicity one. It is implemented by a unitary
\begin{equation}U:\HH_A\longrightarrow\HH_B\end{equation}
such that
\begin{equation}\pi_\rho(X)=UXU^\dagger,\qquad X\in\B(\HH_A).\end{equation}
Its commutant is therefore
\begin{equation}[\pi_\rho(\B(\HH_A))]'=\C I_B.\end{equation}
Together with the normalization $\Tr\rho_B=1$, this gives
\begin{equation}\rho_B=I_B/d.\end{equation}
Since $\pi_\rho=\Phi_\rho|_{\B(\HH_A)}$, we obtain
\begin{equation}\Lambda_\rho(X)=\rho_B^{1/2}\pi_\rho(X)\rho_B^{1/2}=\tfrac1d UXU^\dagger,\qquad X\in\B(\HH_A).\end{equation}
Now take an orthonormal basis $\{\ket i\}_{i=0}^{d-1}$ of $\HH_A$ and write
\begin{equation} E_{ij}=\ket i\!\bra j .\end{equation}
The bipartite state $\rho$ can be expanded as
\begin{equation}\rho=\sum_{i,j=0}^{d-1}E_{ij}\otimes R_{ij},\end{equation}
where
\begin{equation}R_{ij}:=(\bra i\otimes I_B)\rho(\ket j\otimes I_B)\in\B(\HH_B).\end{equation}
Equivalently, $E_{ij}^T=E_{ji}$ gives
\begin{equation}R_{ij}=\Tr_A[(E_{ij}^T\otimes I_B)\rho].\end{equation}
By the definition of $\Lambda_\rho$, it follows that
\begin{equation} R_{ij}=\Lambda_\rho(E_{ij})=\tfrac1d UE_{ij}U^\dagger.\end{equation}
Hence
\begin{equation}\rho=\sum_{i,j=0}^{d-1}E_{ij}\otimes R_{ij}=\tfrac1d\sum_{i,j=0}^{d-1}E_{ij}\otimes UE_{ij}U^\dagger.\end{equation}
The maximally entangled state has the expansion
\begin{equation}\proj{\psi_d}=\tfrac1d\sum_{i,j=0}^{d-1}E_{ij}\otimes E_{ij}.\end{equation}
It follows that
\begin{equation}\rho=(I_A\otimes U)\proj{\psi_d}(I_A\otimes U^\dagger).\end{equation}
Thus $\rho$ and $\proj{\psi_d}$ are related by a local unitary on Bob's side. Up to local unitary equivalence,
\begin{equation}\rho\simeq\proj{\psi_d}.\end{equation}
This proves the corollary.
\end{proof}

\section{Maximal spectral gap for $d=3$ rank-one binary settings}\label{app:qutrit}
We first prove Theorem~\ref{maxgap}, then give the general structure of choices that attain
this bound, and briefly discuss the case $d>3$.
\begin{proof}[Proof of Theorem~\ref{maxgap}]
We first prove $\nu_{\mathrm{opt}}\le (3-\sqrt3)/6$, and then show a concrete choice attaining the gap $\nu_*= (3-\sqrt3)/6$.

\paragraph{A variational expression for the gap.}

For any operator $S,T\in M_3(\C)$, define vectorization and the Hilbert--Schmidt inner product by
\begin{equation}
  \vct(T)=\sum_{a,b=1}^3T_{ab}\ket{ab},\qquad
  \langle S,T\rangle_{\mathrm{HS}}=\Tr(S^\dagger T).
\end{equation}
We write $\|T\|_2^2=\Tr(T^\dagger T)$. For a rank-one orthogonal projection $Q$,
the pass operator for a single binary setting is
\begin{equation}
  \Omega_Q=Q\otimes Q^T+(I-Q)\otimes(I-Q)^T.
\end{equation}
The vectorization identity $(A\otimes B)\vct(T)=\vct(ATB^T)$ and $Q^2=Q$ give us
\begin{align}
  (I-\Omega_Q)\vct(T)
    &=\vct\bigl(QT(I-Q)+(I-Q)TQ\bigr)\notag\\
    &=\vct\bigl([Q,[Q,T]]\bigr).
    \label{eq:one-reject}
\end{align}
Define the superoperators on matrix space by
\begin{equation}
  D_QT:=[Q,[Q,T]],\qquad
  \mathcal L:=\sum_{i=1}^3\mu_iD_{Q_i}.
\end{equation}
Since $\Omega=\sum_{i=1}^3\mu_i\Omega_{Q_i}$ and $\sum_{i=1}^3\mu_i=1$,
Eq.~\eqref{eq:one-reject} and linearity of vectorization give
\begin{equation*}
\begin{aligned}
  (I-\Omega)\vct(T)
  &=\sum_{i=1}^3\mu_i(I-\Omega_{Q_i})\vct(T)\\
  &=\sum_{i=1}^3\mu_i\,\vct(D_{Q_i}T)\\
  &=\vct\!\left(\sum_{i=1}^3\mu_i D_{Q_i}T\right)=\vct(\mathcal LT).
\end{aligned}
\tag{\theequation a}\label{eq:weighted-reject}
\end{equation*}
We also have $D_Q^2=D_Q=D_Q^\dagger$. Put $R=I-Q$, so that $D_QT=QTR+RTQ$.
Using $Q^2=Q$, $R^2=R$, and $QR=RQ=0$, we obtain
\begin{equation*}
\begin{aligned}
  D_Q^2T
  &=Q(QTR+RTQ)R+R(QTR+RTQ)Q\\
  &=Q^2TR^2+R^2TQ^2=QTR+RTQ=D_QT.
\end{aligned}
\tag{\theequation b}\label{eq:reject-idempotent}
\end{equation*}
For any $S,T\in M_3(\C)$, cyclicity of the trace gives
\begin{equation*}
\begin{aligned}
  \langle S,D_QT\rangle_{\mathrm{HS}}
  &=\Tr(S^\dagger QTR)+\Tr(S^\dagger RTQ)\\
  &=\Tr\!\left((RS^\dagger Q+QS^\dagger R)T\right)\\
  &=\Tr\!\left((D_QS)^\dagger T\right)=\langle D_QS,T\rangle_{\mathrm{HS}}.
\end{aligned}
\tag{\theequation c}\label{eq:reject-selfadjoint}
\end{equation*}
Thus $D_Q$ is a Hilbert--Schmidt orthogonal projection.
Since
$[Q,T]=QT(I-Q)-(I-Q)TQ$, we have
\begin{align}
  \langle T,D_QT\rangle_{\mathrm{HS}}
    &=\|D_QT\|_2^2\notag\\
    &=\|QT(I-Q)\|_2^2+\|(I-Q)TQ\|_2^2\notag\\
    &=\|[Q,T]\|_2^2.
\end{align}
Consequently, $\mathcal L$ is self-adjoint and positive semidefinite, and
\begin{equation}
  F(T):=\langle T,\mathcal LT\rangle_{\mathrm{HS}}
  =\sum_{i=1}^3\mu_i\|[Q_i,T]\|_2^2.
\end{equation}
Since $\mu_i> 0$, all terms in the sum are nonnegative. Hence
\begin{equation}
  \ker\mathcal L
  =\bigl\{T:[Q_i,T]=0\
  \text{for every }i\text{ with }\mu_i>0\bigr\}.
\end{equation}
In particular, $\mathcal L(I)=0$. We can also derive the following relation by \ref{thm:verify}

\begin{equation}
  \nu>0
  \ \Longleftrightarrow\ \ker\mathcal L=\C I
  \ \Longleftrightarrow\ \ker(I-\Omega)=\C\ket{\psi_3}.
\end{equation}

So the Rayleigh--Ritz variational principle gives us
\begin{align}
  \lambda_2(\Omega)
    &=\max_{\substack{
        \ket{\phi}\perp\ket{\psi_3}\\
        \langle\phi\mid\phi\rangle=1}}
      \langle\phi\mid\Omega\mid\phi\rangle\notag\\
    &=\max_{\substack{\Tr T=0\\\|T\|_2=1}}
      \langle\vct(T),\Omega\vct(T)\rangle.
\end{align}
It follows that
\begin{align}
  \nu
    &=1-\lambda_2(\Omega)\notag\\
    &=\min_{\substack{\Tr T=0\\\|T\|_2=1}}
      \langle T,\mathcal LT\rangle_{\mathrm{HS}}\notag\\
    &=\min_{\substack{\Tr T=0\\\|T\|_2=1}}F(T).
    \label{eq:rayleigh}
\end{align}
Below we write $D_i=D_{Q_i}$ and omit the subscript $\mathrm{HS}$ on matrix inner products.

\Needspace{5\baselineskip}
\paragraph{A tentative construction}

Relabel the settings so that $\mu_1=\min_i\mu_i\le 1/3$. Choose a unit
vector $u\perp q_2,q_3$ and put $U=\proj{u}$. The traceless, normalized matrix
\begin{equation}
 B=\frac{I-3U}{\sqrt6}
\end{equation}
commutes with $Q_2,Q_3$, hence $D_2B=D_3B=0$. Write $P=D_1$ and
$p=\Tr(Q_1U)=|\braket{q_1}{u}|^2$. The rank-one identity
$\|[Q_1,U]\|_2^2=2p(1-p)$ gives
\begin{equation}
 a:=\langle B,PB\rangle=\|[Q_1,B]\|_2^2
   =3p(1-p)\le\frac34.
\end{equation}
Since $\mathcal LB=\mu_1 PB$, Eq.~\eqref{eq:rayleigh} yields the first bound
\begin{equation}\label{eq:firsttrial}
 \nu\le F(B)=\mu_1 a\le\frac14.
\end{equation}

\Needspace{5\baselineskip}
\paragraph{Refine the first bound.}

To improve the first bound in Eq.~\eqref{eq:firsttrial}, we ask which perturbation of $B$ decreases $F$ most rapidly. Let $X$ be a traceless unit matrix orthogonal to $B$, and for a real parameter $t$, we evaluate
\begin{equation}
  T_t=\frac{B+tX}{\sqrt{1+t^2}},
\end{equation}
and
\begin{equation}
  \left.\frac{\mathrm d}{\mathrm dt}F(T_t)\right|_{t=0}
  =2\operatorname{Re}\langle X,\mathcal LB\rangle.
\end{equation}
Since $\mathcal LB=\mu_1 PB$ and $X\perp B$, we have
\begin{equation}\label{eq:slope}
  \begin{aligned}
    \left.\frac{\mathrm d}{\mathrm dt}F(T_t)\right|_{t=0}
    &=2\mu_1\operatorname{Re}\langle X,PB\rangle\\
    &=2\mu_1\operatorname{Re}\langle X,PB-aB\rangle.
  \end{aligned}
\end{equation}
If $a=0$, then $PB=0$ and $\mathcal LB=0$, so Eq.~\eqref{eq:firsttrial} gives $\nu=0$ and the upper bound holds. We may therefore assume $0<a\le3/4$.
By the Cauchy--Schwarz inequality, the steepest descent direction is
\begin{equation}
  X=-\frac{PB-aB}{\|PB-aB\|_2}.
\end{equation}
Because $P$ is an orthogonal projection,
\begin{equation}
  \|PB-aB\|_2^2=a(1-a),
\end{equation}
we use the following Gram--Schmidt orthogonalization of $PB$ against $B$ 
\begin{equation}
  C=\frac{PB-aB}{\sqrt{a(1-a)}},
\end{equation}
and consider the two-dimensional trial space
\begin{equation}
  S=\operatorname{span}\{B,C\}
   =\operatorname{span}\{B,\mathcal LB\}.
\end{equation}

Take
\begin{equation}
  T=xB+yC,\qquad |x|^2+|y|^2=1.
\end{equation}
Let $s=\sqrt{a(1-a)}$, and $P^2=P$ give
\begin{equation}
  PB=aB+sC,\qquad PC=sB+(1-a)C.
\end{equation}
Consequently,
\begin{equation}
  \|PT\|_2^2
  =\bigl|\sqrt a\,x+\sqrt{1-a}\,y\bigr|^2.
\end{equation}
On the other hand, $D_2B=D_3B=0$, and orthogonal projections do not increase the norm. Thus
\begin{equation}
  \begin{aligned}
    \mu_2\|D_2T\|_2^2+\mu_3\|D_3T\|_2^2
    &=|y|^2\sum_{j=2}^3\mu_j\|D_jC\|_2^2\\
    &\le(1-\mu_1)|y|^2.
  \end{aligned}
\end{equation}
Combining the three settings yields
\begin{equation}
  F(xB+yC)
  \le\mu_1\bigl|\sqrt a\,x+\sqrt{1-a}\,y\bigr|^2
      +(1-\mu_1)|y|^2.
\end{equation}
The right-hand side is a quadratic form:
\begin{equation}
  F(xB+yC)
  \le
  \begin{pmatrix}\bar x&\bar y\end{pmatrix}
  M
  \begin{pmatrix}x\\y\end{pmatrix},
\end{equation}
where
\begin{equation}
  M=\begin{pmatrix}
    \mu_1 a&\mu_1 s\\
    \mu_1 s&1-\mu_1 a
  \end{pmatrix}.
\end{equation}
We can check
\begin{equation}
  \operatorname{Tr}M=1,\qquad
  \det M=\mu_1(1-\mu_1)a.
\end{equation}
Therefore,
\begin{equation}
  \lambda_{\min}(M)
  =\frac{1-\sqrt{1-4\mu_1(1-\mu_1)a}}{2}.
\end{equation}
Using the bounds established above,
\begin{equation}
  \mu_1\le\frac13,\qquad a\le\frac34,
\end{equation}
we obtain
\begin{equation}
  4\mu_1(1-\mu_1)a\le\frac23.
\end{equation}
Hence
\begin{equation}\label{eq:qutrit-upper}
  \nu\le\lambda_{\min}(M)
  \le\frac{1-1/\sqrt3}{2}
  =\frac{3-\sqrt3}{6}.
\end{equation}

\paragraph{Attaining the bound}
We then offer a direct construction attaining the previous upper bound.
Take $\mu_1=\mu_2=\mu_3=1/3$ and $\omega=e^{2\pi\mathrm i/3}$. We use the three directions from Eq.~\eqref{eq:qutrit-design},
\begin{equation}
 q_1=\frac{(1,1,0)^T}{\sqrt2},\qquad
 q_2=\frac{(0,1,1)^T}{\sqrt2},\qquad
 q_3=\frac{(0,1,\omega)^T}{\sqrt2}.
\end{equation}
Under the given basis, the matrix of superoperator $\mathcal L$ is
\begin{equation}
 [\mathcal L]=I_9-\Omega
 =\frac13\sum_{i=1}^3
 \left(Q_i\otimes I_3+I_3\otimes Q_i^T-2Q_i\otimes Q_i^T\right).
\end{equation}
Substituting the three projectors $Q_i=\proj{q_i}$ and using
$1+\omega+\omega^2=0$ gives us the characteristic polynomial
\begin{equation}
 \det\!\bigl(tI_9-[\mathcal L]\bigr)
 =t\left(t-\tfrac12\right)^2\left(t^2-t+\tfrac16\right)^3 .
\end{equation}
Writing $\lambda_\pm=(3\pm\sqrt3)/6$, we have
\begin{equation}
 \operatorname{spec}(I_9-\Omega)
 =\bigl\{0,\ \lambda_-^{\times3},\ (1/2)^{\times2},\ \lambda_+^{\times3}\bigr\}.
\end{equation}
As we expect, the zero eigenvalue is simple, so the maximal gap is $\nu_\mathrm{opt}=\lambda_-=(3-\sqrt3)/6$.
Thus these settings attain the upper bound in Eq.~\eqref{eq:qutrit-upper}.
\end{proof}

More generally, we have the following theorem.

\begin{theorem}\label{innerp}
Let $v_1,v_2,v_3\in\C^3$ be unit vectors, $V_i=\proj{v_i}$ and
$F(T)=\sum_i\mu_i\|[V_i,T]\|_2^2$. Denote $g_{ij}=\braket{v_i}{v_j}$. The bound in Eq.~\eqref{eq:qutrit-upper}
is attained if and only if $\mu_1=\mu_2=\mu_3=1/3$ and
\begin{equation}\label{eq:qutrit-pair-condition}
 |g_{ij}|^2=\frac14\qquad(i\ne j),
\end{equation}
\begin{equation}\label{eq:qutrit-cycle-condition}
 \operatorname{Re}(g_{12}g_{23}g_{31})=\frac1{16}.
\end{equation}
\end{theorem}

\begin{proof}

\textbf{Sufficiency.}
Let $\mu_i=1/3$ and assume Eqs.~\eqref{eq:qutrit-pair-condition}--\eqref{eq:qutrit-cycle-condition}.
For $z=g_{12}g_{23}g_{31}$, we have $|z|=1/8$ and $\operatorname{Re}z=1/16$, hence
\begin{equation}
 z=\frac{1\pm\mathrm i\sqrt3}{16}.
\end{equation}
Interchanging $v_2,v_3$ conjugates $z$ and rephasing $v_2,v_3$ leaves the projectors
unchanged, so we may take $z=(1+\mathrm i\sqrt3)/16$ and $g_{12}=g_{13}=1/2$.
Consequently, $g_{23}=4z=(1+\omega)/2$, and
\begin{equation}
 (g_{ij})_{i,j=1}^3
 =\bigl(\braket{q_i}{q_j}\bigr)_{i,j=1}^3
 =\frac12\begin{pmatrix}
 2&1&1\\
 1&2&1+\omega\\
 1&1+\bar\omega&2
 \end{pmatrix}.
\end{equation}
Since $q_1,q_2,q_3$ form a basis, equality of the Gram matrices gives a
unitary $W$ with $Wq_i=v_i$. Thus
\begin{equation}
 \Omega_v=(W\otimes\bar W)\Omega(W\otimes\bar W)^\dagger.
\end{equation}
Therefore, $\Omega_v$ attains the same gap as $\Omega$.

\textbf{Necessity.}
Let $\nu_*=(3-\sqrt3)/6$ and suppose $\nu=\nu_*$.
Relabel so that $\mu_1=\min_i\mu_i$, and take a unit vector $u\perp v_2,v_3$.
Set $p=|\braket{u}{v_1}|^2$ and $a=3p(1-p)$.
Equality in Eq.~\eqref{eq:qutrit-upper} yields
\begin{equation}
 \frac23=4\mu_1(1-\mu_1)a
 \le4\cdot\frac29\cdot\frac34=\frac23.
\end{equation}
Hence $\mu_i=1/3$ for all $i$, and $p=1/2$.

Put $e=\sqrt2\bigl(v_1-\braket{u}{v_1}u\bigr)$, and complete $\{u,e\}$
to an orthonormal basis $\{u,e,f\}$. For $j=2,3$, set
$x_j=|\braket{e}{v_j}|^2$; note that $v_j\in\operatorname{span}\{e,f\}$.
Choose the traceless Hermitian test matrix
\begin{equation}
 T=\operatorname{diag}(-1-\sqrt3,\,\sqrt3-1,\,2)
\end{equation}
in this basis, so that $\|T\|_2^2=12$.
Using $\|[\proj{v},T]\|_2^2
=2\bigl(\bra{v}T^2\ket{v}-(\bra{v}T\ket{v})^2\bigr)$,
Eq.~\eqref{eq:rayleigh} gives
\begin{equation}\label{eq:qutrit-diagonal-energy}
 \begin{aligned}
 \nu_*\le\frac{F(T)}{12}
 &=\frac16+\frac{2-\sqrt3}{3}\sum_{j=2}^3 x_j(1-x_j)\\
 &=\nu_*-\frac{2-\sqrt3}{3}\sum_{j=2}^3
 \left(x_j-\frac12\right)^2.
 \end{aligned}
\end{equation}
Thus $x_2=x_3=1/2$, and
\begin{equation}
 |g_{1j}|^2=(1-p)x_j=\frac14,\qquad j=2,3.
\end{equation}
Repeating the argument with $v_2$ first yields $|g_{23}|^2=1/4$.

Finally, let $G=(g_{ij})_{i,j=1}^3$ and $z=g_{12}g_{23}g_{31}$.
Since $v_2,v_3$ span $u^\perp$, Gram--Schmidt gives
\begin{equation}
 \det G=(1-|g_{23}|^2)p=\frac38.
\end{equation}
Direct expansion also gives
\begin{equation}
 \det G=1-\sum_{i<j}|g_{ij}|^2+2\operatorname{Re}z
 =\frac14+2\operatorname{Re}z.
\end{equation}
Therefore $\operatorname{Re}(g_{12}g_{23}g_{31})=1/16$, as required.
\end{proof}

\paragraph{An upper bound in higher dimensions}\label{app:high-d-gap}
We first briefly prove, for $d>3$, the minimum number of rank-one binary settings needed to verify $\ket{\psi_d}$
uniquely is $d$ (see also Proposition~2.1 and Theorem~1.1 of Wright~\cite{Wright2011}).
With fewer than $d$ tested rays, there is a unit vector $u$ orthogonal to them all. The product
state $\ket{u}\otimes\ket{\bar u}$ then passes every setting, since both parties always obtain
the complementary outcome. To see that $d$ settings suffice, take $q_1=e_1$ and
$q_j=(e_1+e_j)/\sqrt2$ for $2\le j\le d$ in the standard basis. The $q_j$ form a basis with
pairwise nonzero overlaps, so the products $Q_iQ_j=\braket{q_i}{q_j}\ket{q_i}\!\bra{q_j}$
span $M_d(\C)$. By Eq.~\eqref{eq:binarycriterion}, these settings uniquely verify the target
whenever each is chosen with positive probability.

To bound the spectral gap with this minimum number of settings, write the rejection operator
for a single setting as
\begin{equation}\label{eq:high-d-reject}
 I-\Omega_Q=Q\otimes(I-Q)^T+(I-Q)\otimes Q^T.
\end{equation}
The two terms have orthogonal supports, each of rank $d-1$.
Hence, for any mixture of $d$ such settings with $\sum_i\mu_i=1$,
\begin{equation}\label{eq:high-d-trace}
 \Tr(I-\Omega)
 =\sum_{i=1}^d\mu_i\Tr(I-\Omega_{Q_i})
 =2(d-1).
\end{equation}
Since the target state is a zero eigenvector of $I-\Omega$,
the smallest eigenvalue on its $(d^2-1)$-dimensional orthogonal
complement cannot exceed the average eigenvalue. Therefore,
\begin{equation}\label{eq:high-d-opt-gap}
 0\le\nu_{\mathrm{opt}}(d)
 \le\frac{2(d-1)}{d^2-1}
 =\frac{2}{d+1}
 =O(d^{-1}).
\end{equation}

\section{Auxiliary calculations}
\subsection{Algebras generated by the example measurements}
\begin{lemma}[The two-qubit coarse-grained projections generate the full algebra]\label{lem:twoqubit}
The three projections $P_Z,P_X,P_R$ of the two-qubit example generate $\B(\C^2\otimes\C^2)$.
\end{lemma}
\begin{proof}
Write
\begin{equation}E_{ij}=\ket i\!\bra j,\qquad i,j\in\{0,1\},\end{equation}
and let
\begin{equation}\A=C^*(P_Z,P_X,P_R).\end{equation}
We prove that $\A=\B(\C^2\otimes\C^2)$. First,
\begin{equation}P_Z=\proj0\otimes I=E_{00}\otimes I\in\A,\end{equation}
so
\begin{equation}E_{11}\otimes I=I-E_{00}\otimes I\in\A.\end{equation}
Also,
\begin{equation}P_X=\proj+\otimes I=\tfrac12(E_{00}+E_{01}+E_{10}+E_{11})\otimes I.\end{equation}
Hence
\begin{equation}(E_{00}\otimes I)P_X(E_{11}\otimes I)=\tfrac12 E_{01}\otimes I,\end{equation}
which yields
\begin{equation}E_{01}\otimes I\in\A.\end{equation}
Similarly,
\begin{equation}(E_{11}\otimes I)P_X(E_{00}\otimes I)=\tfrac12 E_{10}\otimes I,\end{equation}
so
\begin{equation}E_{10}\otimes I\in\A.\end{equation}
Together these inclusions give
\begin{equation}E_{ij}\otimes I\in\A,\qquad i,j\in\{0,1\}.\end{equation}
Equivalently,
\begin{equation}\B(\C^2)\otimes I\subseteq\A.\end{equation}
We now use the definition of $P_R$,
\begin{equation}P_R=\proj0\otimes\proj++\proj1\otimes\proj0.\end{equation}
Its diagonal blocks satisfy
\begin{equation}(E_{00}\otimes I)P_R(E_{00}\otimes I)=E_{00}\otimes\proj+,\end{equation}
and
\begin{equation}(E_{11}\otimes I)P_R(E_{11}\otimes I)=E_{11}\otimes\proj0.\end{equation}
Since all $E_{ij}\otimes I$ belong to $\A$, for any $i,j\in\{0,1\}$ we obtain
\begin{equation}(E_{i0}\otimes I)(E_{00}\otimes\proj+)(E_{0j}\otimes I)=E_{ij}\otimes\proj+\in\A,\end{equation}
and
\begin{equation}(E_{i1}\otimes I)(E_{11}\otimes\proj0)(E_{1j}\otimes I)=E_{ij}\otimes\proj0\in\A.\end{equation}
Taking $i=j$ and summing over $i=0,1$ gives
\begin{equation}I\otimes\proj+\in\A,\end{equation}
and
\begin{equation}I\otimes\proj0\in\A.\end{equation}
We also have
\begin{equation}I\otimes\proj1=I-I\otimes\proj0\in\A.\end{equation}
Since
\begin{equation}\proj+=\tfrac12(\proj0+\ket0\!\bra1+\ket1\!\bra0+\proj1),\end{equation}
we can form the products
\begin{equation}(I\otimes\proj0)(I\otimes\proj+)(I\otimes\proj1)=\tfrac12 I\otimes\ket0\!\bra1,\end{equation}
and
\begin{equation}(I\otimes\proj1)(I\otimes\proj+)(I\otimes\proj0)=\tfrac12 I\otimes\ket1\!\bra0.\end{equation}
Consequently,
\begin{equation}I\otimes E_{ij}\in\A,\qquad i,j\in\{0,1\},\end{equation}
which implies
\begin{equation}I\otimes\B(\C^2)\subseteq\A.\end{equation}
Finally, for all $i,j,k,l\in\{0,1\}$,
\begin{equation}(E_{ij}\otimes I)(I\otimes E_{kl})=E_{ij}\otimes E_{kl}\in\A.\end{equation}
These matrix units span $M_4(\C)$, so
\begin{equation}\A=\B(\C^2\otimes\C^2).\end{equation}
This proves the lemma.
\end{proof}

\Needspace{5\baselineskip}
\begin{lemma}[The computational and tilted bases generate the full qubit algebra]\label{lem:tilted}
For $0<\theta<\pi/2$, the projections $\proj0$ and $\proj{0_\theta}$ generate $M_2(\C)$.
\end{lemma}
\begin{proof}
Set
\begin{equation}\A_\theta=C^*(\proj0,\proj{0_\theta}).\end{equation}
By the finite-dimensional double commutant theorem~\cite{Davidson1996}, it suffices to show that
\begin{equation}\A_\theta'=\C I.\end{equation}
Take any $T\in\A_\theta'$. Since $T$ commutes with $\proj0$, it is diagonal in the computational basis,
\begin{equation}T=\begin{pmatrix}a&0\\0&b\end{pmatrix}.\end{equation}
It also commutes with $\proj{0_\theta}$, so the one-dimensional subspace spanned by $\ket{0_\theta}$ is
invariant under $T$. Thus there is a $\lambda\in\C$ such that
\begin{equation} T\ket{0_\theta}=\lambda\ket{0_\theta}.\end{equation}
Using
\begin{equation}\ket{0_\theta}=\cos\theta\,\ket0+\sin\theta\,\ket1\end{equation}
gives
\begin{equation}a\cos\theta\,\ket0+b\sin\theta\,\ket1=\lambda\cos\theta\,\ket0+\lambda\sin\theta\,\ket1.\end{equation}
Both $\cos\theta$ and $\sin\theta$ are nonzero, so $a=b=\lambda$. Thus every such $T$ satisfies
\begin{equation}T=\lambda I.\end{equation}
We conclude that
\begin{equation}\A_\theta'=\C I,\end{equation}
which gives
\begin{equation}\A_\theta=M_2(\C).\end{equation}
\end{proof}

\subsection{Calculation of the Bayes error}
For the state $\ket{\psi_\lambda}=\sqrt\lambda\ket{00}+\sqrt{1-\lambda}\ket{11}$ and the tilted basis
$\ket{0_\theta}=\cos\theta\,\ket0+\sin\theta\,\ket1$,
$\ket{1_\theta}=-\sin\theta\,\ket0+\cos\theta\,\ket1$, the unnormalized conditional memory vectors are
\begin{align}
\ket{\beta_0}&=\sqrt\lambda\cos\theta\ket0+\sqrt{1-\lambda}\sin\theta\ket1,\\
\ket{\beta_1}&=-\sqrt\lambda\sin\theta\ket0+\sqrt{1-\lambda}\cos\theta\ket1 .
\end{align}
Set $\sigma_i=\proj{\beta_i}$. The binary Helstrom error is
\begin{equation} e_{\rm Bayes}=\tfrac12(1-\|\sigma_0-\sigma_1\|_1).\end{equation}
For two unnormalized rank-one states whose traces sum to one,
\begin{equation}\|\sigma_0-\sigma_1\|_1=\sqrt{1-4|\braket{\beta_0}{\beta_1}|^2}.\end{equation}
Here $\braket{\beta_0}{\beta_1}=(1-2\lambda)\sin\theta\cos\theta$, so
\begin{equation}e_{\rm Bayes}=\tfrac12\left[1-\sqrt{1-(1-2\lambda)^2\sin^2(2\theta)}\right].\end{equation}

\begin{proposition}[No uniform distance bound from the Bayes error]\label{prop:bayes}
For a state $\rho$ and a measurement family $\M$, write
$\epsilon_{\rm Bayes}(\rho;\M)=\max_x e_{\rm Bayes}^{(x)}(\rho)$ for the largest Bayes error over the
settings. Consider all two-setting qubit tasks whose measurements generate the full qubit algebra, with
a memory of minimal dimension. There is no function $f(\epsilon)\to0$ as $\epsilon\to0$ such that, for
every such task and every shared state, $\epsilon_{\rm Bayes}(\rho;\M)\le\epsilon$ implies
\begin{equation}\inf_{\omega\in\mathcal E_\M}\tfrac12\|\rho-\omega\|_1\le f(\epsilon),\end{equation}
where $\mathcal E_\M$ denotes the set of ZUSs for $\M$.
\end{proposition}
\begin{proof}
Consider a family of two-setting qubit tasks $\M_\theta$ consisting of the computational-basis PVM and a tilted-basis PVM,
\begin{align}
\mathcal Z&=\{\proj0,\proj1\},\\
\mathcal T_\theta&=\{\proj{0_\theta},\proj{1_\theta}\},\qquad 0<\theta<\pi/4,
\end{align}
where
\begin{equation}\ket{0_\theta}=\cos\theta\,\ket0+\sin\theta\,\ket1,\qquad
\ket{1_\theta}=-\sin\theta\,\ket0+\cos\theta\,\ket1.\end{equation}
By Lemma~\ref{lem:tilted}, these two projective measurements generate $M_2(\C)$ and hence probe the entire qubit.

Fix the minimum memory dimension $\dim\HH_B=2$ and take the bipartite pure state
\begin{equation}\ket{\psi_\theta}=\ket{00},\qquad \rho_\theta=\proj{00}.\end{equation}
It is the member of the family
\begin{equation}\ket{\psi_\lambda}=\sqrt\lambda\,\ket{00}+\sqrt{1-\lambda}\,\ket{11}\end{equation}
with $\lambda=1$.

We first calculate $\epsilon_{\rm Bayes}(\rho_\theta;\M_\theta)$. In the computational-basis setting
$\mathcal Z$, Alice obtains outcome $0$ with certainty, so the minimum Bayes error is
\begin{equation}e_{\rm Bayes}^{(\mathcal Z)}(\rho_\theta)=0.\end{equation}
For the tilted setting $\mathcal T_\theta$, the Helstrom formula from the main text gives
\begin{equation}e_{\rm Bayes}^{(\mathcal T_\theta)}(\rho_\theta)=\tfrac12\bigl(1-\sqrt{1-\sin^2(2\theta)}\bigr).\end{equation}
Since $\cos(2\theta)>0$ for $0<\theta<\pi/4$, this reduces to
\begin{equation}e_{\rm Bayes}^{(\mathcal T_\theta)}(\rho_\theta)=\tfrac12\bigl(1-\cos(2\theta)\bigr)=\sin^2\theta .\end{equation}
Thus
\begin{equation}
\begin{split}
\epsilon_{\rm Bayes}(\rho_\theta;\M_\theta)&=\max\{e_{\rm Bayes}^{(\mathcal Z)}(\rho_\theta),e_{\rm Bayes}^{(\mathcal T_\theta)}(\rho_\theta)\}\\
&=\sin^2\theta\longrightarrow0\quad(\theta\to0).
\end{split}
\end{equation}
On the other hand, the full-algebra case of Theorem~\ref{thm:representation} shows that when $\M_\theta$
generates $M_2(\C)$ and the memory has the minimum dimension two, the ZUSs are precisely the maximally
entangled two-qubit states. Hence
\begin{equation}
\begin{split}
\mathcal E_{\M_\theta}&=\{\omega=\proj{\psi_{\rm ME}}:\\
&\quad\ket{\psi_{\rm ME}}\ \text{is a maximally entangled two-qubit state}\}.
\end{split}
\end{equation}
We now bound the trace distance from $\rho_\theta=\proj{00}$ to these states. Take any
$\omega=\proj{\psi_{\rm ME}}\in\mathcal E_{\M_\theta}$. Since $\ket{\psi_{\rm ME}}$ is a maximally
entangled two-qubit state, it has the Schmidt form
\begin{equation}\ket{\psi_{\rm ME}}=\tfrac1{\sqrt2}\bigl(\ket{a_0}\ket{b_0}+e^{i\varphi}\ket{a_1}\ket{b_1}\bigr),\end{equation}
where $\{\ket{a_0},\ket{a_1}\}$ and $\{\ket{b_0},\ket{b_1}\}$ are orthonormal bases. It follows that
\begin{equation}\braket{00}{\psi_{\rm ME}}=\tfrac1{\sqrt2}\bigl(\braket0{a_0}\braket0{b_0}+e^{i\varphi}\braket0{a_1}\braket0{b_1}\bigr).\end{equation}
The Cauchy--Schwarz inequality gives
\begin{equation}
\begin{split}
|\braket{00}{\psi_{\rm ME}}|^2&\le\tfrac12
\bigl(|\braket0{a_0}|^2+|\braket0{a_1}|^2\bigr)\\
&\qquad\bigl(|\braket0{b_0}|^2+|\braket0{b_1}|^2\bigr).
\end{split}
\end{equation}
By orthonormality of the two bases,
\begin{equation}
|\braket0{a_0}|^2+|\braket0{a_1}|^2=1,\qquad
|\braket0{b_0}|^2+|\braket0{b_1}|^2=1.
\end{equation}
Hence
\begin{equation}|\braket{00}{\psi_{\rm ME}}|^2\le\tfrac12.\end{equation}
For two pure states, the trace distance is
\begin{equation}\tfrac12\|\rho_\theta-\omega\|_1=\sqrt{1-|\braket{00}{\psi_{\rm ME}}|^2}.\end{equation}
Combining these expressions yields
\begin{equation}\tfrac12\|\rho_\theta-\omega\|_1\ge\sqrt{1-\tfrac12}=\tfrac1{\sqrt2}.\end{equation}
Since $\omega\in\mathcal E_{\M_\theta}$ was arbitrary,
\begin{equation}\inf_{\omega\in\mathcal E_{\M_\theta}}\tfrac12\|\rho_\theta-\omega\|_1\ge\tfrac1{\sqrt2}.\end{equation}
Suppose that a function $f(\epsilon)\to0$ gave the uniform bound in the proposition. Applying it to the
task $\M_\theta$ and the state $\rho_\theta$ would give
\begin{equation}
\tfrac1{\sqrt2}\le\inf_{\omega\in\mathcal E_{\M_\theta}}\tfrac12\|\rho_\theta-\omega\|_1
\le f(\epsilon_{\rm Bayes}(\rho_\theta;\M_\theta))=f(\sin^2\theta).
\end{equation}
As $\theta\to0$, the right-hand side tends to zero while the left-hand side is the positive constant
$1/\sqrt2$. This contradiction rules out such a uniform bound.
\end{proof}

\subsection{Calculation for the matching-outcome trap test}
Conditional on the honest report $0_\theta$, the normalized memory state is
$\hat\sigma_0=\sigma_0/\Tr\sigma_0$. The rejection probability is therefore
\begin{equation}e_{\rm false}=1-\Tr\bigl(\proj{0_\theta}\hat\sigma_0\bigr).\end{equation}
From the expression for $\ket{\beta_0}$, we have
$\Tr\sigma_0=\lambda\cos^2\theta+(1-\lambda)\sin^2\theta$ and
$|\braket{0_\theta}{\beta_0}|^2=(\sqrt\lambda\cos^2\theta+\sqrt{1-\lambda}\sin^2\theta)^2$. Substituting
both expressions gives
\begin{equation}e_{\rm false}=\frac{\sin^2\theta\cos^2\theta\,(\sqrt\lambda-\sqrt{1-\lambda})^2}
{\lambda\cos^2\theta+(1-\lambda)\sin^2\theta}.\end{equation}

\bibliographystyle{unsrtnat}
\bibliography{references}
\end{document}